\documentclass[10pt,letterpaper,conference]{IEEEtran}
\IEEEoverridecommandlockouts

\usepackage{graphicx}
\usepackage{amsmath,amsxtra, amsfonts, amssymb}
\usepackage{setspace}
\usepackage{times}
\usepackage{epsfig}
\usepackage{subfigure}
\usepackage{latexsym}
\usepackage{cite}
\usepackage{url}
\usepackage{epstopdf}
\usepackage{tikz,pgf,pgfplots}
\usepackage{multirow}
\usepackage{cool} 
\usepackage{flushend}

\usepackage{amsthm}

\newtheorem{theorem}{Theorem}
\newtheorem{lemma}{Lemma}
\newtheorem{proposition}{Proposition}

\newtheorem{corollary}{Corollary}
\newtheorem{definition}{Definition}

\newcommand{\Psr}{\mathcal{P}_{\rm{sr}}}
\newcommand{\Prd}{\mathcal{P}_{\rm{rd}}}
\newcommand{\Psrbar}{\overline{\mathcal{P}}_{\rm{sr}}}
\newcommand{\Prdbar}{\overline{\mathcal{P}}_{\rm{rd}}}
\newcommand{\Pout}{\mathcal{P}_{\rm{out}}}

\begin{document}
\title{Improper Gaussian Signaling in Full-Duplex Relay Channels with Residual Self-Interference\vspace{-4mm}}
\author{\IEEEauthorblockN{Mohamed Gaafar, Mohammad Galal Khafagy, Osama Amin, and Mohamed-Slim Alouini
\thanks{The research reported in this publication was supported by funding from King Abdullah University of Science and Technology (KAUST).}}\\
\IEEEauthorblockA{Computer, Electrical, and Mathematical Sciences and Engineering (CEMSE) Division\\
King Abdullah University of Science and Technology (KAUST), Thuwal, Makkah Province, Saudi Arabia.\\
Email: \{mohamed.gaafar, mohammad.khafagy, osama.amin, slim.alouini\}@kaust.edu.sa\}@kaust.edu.sa
}
}
\maketitle
\pagenumbering{gobble}
\begin{abstract}
We study the potential employment of improper Gaussian signaling (IGS) in full-duplex cooperative settings with residual self-interference (RSI). IGS is recently shown to outperform traditional proper Gaussian signaling (PGS) in several interference-limited channel settings. In this work, IGS is employed in an attempt to alleviate the RSI adverse effect in full-duplex relaying (FDR). To this end, we derive a tight upper bound expression for the end-to-end outage probability in terms of the relay signal parameters represented in its power and circularity coefficient. We further show that the derived upper bound is either monotonic or unimodal in the relay's circularity coefficient. This result allows for easily locating the global optimal point using known numerical methods. Based on the analysis, IGS allows FDR systems to operate even with high RSI. It is shown that, while the communication totally fails with PGS as the RSI increases, the IGS outage probability approaches a fixed value that depends on the channel statistics and target rate. The obtained results show that IGS can leverage higher relay power budgets than PGS to improve the performance, meanwhile it relieves its RSI impact via tuning the signal impropriety.
\end{abstract}

\section{Introduction}
Contrary to a long-held acceptance that radio front-ends cannot simultaneously transmit and receive, a truly promising potential for \emph{full-duplex communications} has been shown by recent hardware developments \cite{mobicom2011fullduplex},\cite{201212_TWC_Duarte}. Indeed, by multiplexing inbound and outbound traffic over the same channel resource, a full-duplex radio can recover the spectral efficiency loss known to be encountered by its half-duplex counterpart. Performance merits of full-duplex radio have been recently investigated in different communication settings, including full-duplex bidirectional communication, full-duplex base stations, and full-duplex relaying (FDR) \cite{201409_JSAC_FD_Tutorial}, with the latter being the focus of this work.  These merits have qualified full-duplex communication to be considered as a candidate technology for future fifth generation (5G) wireless networks \cite{201402_COMMAG_5G_full_duplex}. 

FDR allows a relay node to listen to an information source and simultaneously forward to its intended destination. Theoretically, this simultaneous transmission/reception doubles the spectral efficiency in the relay channel. However, in practice, this comes at the expense of a self-interference level introduced at the receiver of the relay node from its own transmitter. Even with the application of advanced self-interference isolation and cancellation techniques, a level of residual self-interference (RSI) persists. Such a persistent RSI link and the means to mitigate it represent the main challenge in full-duplex communications, especially with the fact that its adverse effect can typically be an increasing function of the relay power. Therefore, increasing the relay power no longer guarantees an enhanced end-to-end performance. For instance, by increasing the relay power in a fixed-rate transmission scheme, the relay may forward more reliably to the destination in the second hop. However, it also increases the RSI level which negatively affects the reliability in the first hop. Hence, higher relay power budgets cannot be always utilized beyond a certain threshold. Consequently, employing interference mitigation strategies in FDR is crucial to accomplish a satisfactory performance of full-duplex transmissions.

Improper Gaussian signaling (IGS) has been recently studied in \cite{cadambe2010interference}, where higher degrees of freedom for the $3$-user single-input single-output interference channel were shown to be achievable. This comes in contrary to other communication settings with interference-free channels where proper Gaussian signaling (PGS) is the optimal choice. PGS assumes the zero-mean complex Gaussian transmit signal to be statistically circularly symmetric with uncorrelated real and imaginary components. On the other hand, IGS is a general class of signals where circularity and uncorrelatedness conditions can be relaxed \cite{schreier2010statistical}. The results in \cite{cadambe2010interference} motivate the need to further study the potential gains of IGS in communication scenarios where interference imposes a noticeable limitation. 

For Gaussian channels, and within the class of Gaussian signals, IGS has been recently adopted to improve the performance of different interference channel communication systems \cite{zeng2013transmit, amin2015outage, Gaafar2015Spectrum, gaafar2015sharing}. In general, recent work on the interference channel showed that IGS can actually support higher rates in certain interference-limited scenarios \cite{zeng2013transmit}. In an underlay cognitive radio setting, IGS increases the spectrum sharing opportunity for secondary users by relieving the interference impact on the authorized users \cite{amin2015outage, Gaafar2015Spectrum, gaafar2015sharing}.  


The potential gains of IGS have been also recently studied in \cite{Hellings2014OnOptimal} for the multiple-input multiple-output (MIMO) relay channel when a partial decode-and-forward strategy is adopted. In such a relaying strategy, the relay only decodes a part of the message, while the rest of the message is treated as an additional interference term. It was shown in \cite{Hellings2014OnOptimal} that PGS is optimal within the class of Gaussian signals. However, the work in \cite{Hellings2014OnOptimal} assumed an ideal full-duplex relay channel, where the self-interference imposed by the relay's transmitter on its own receiver is perfectly canceled.

In this work, we investigate the potential gains of IGS in decode-and-forward FDR with RSI. This work is the first to study IGS in FDR settings with imperfect self-interference cancellation to the best of the authors' knowledge. First, we assume the relay transmits with IGS and derive a closed-form upper bound on the end-to-end outage probability using Jensen's inequality. We compare the derived outage upper bound of IGS to the exact outage probability of PGS from the literature. Also, we show that depending on the channel parameters and transmission rate, the derived upper bound is either monotonic or unimodal in the relay's circularity coefficient, which allows for locating the global optimal value using simple known numerical methods. Through numerical optimization, we show that the use of IGS can yield an outage upper bound that is less than the exact outage probability of PGS. We also numerically validate the derived upper bound by comparing it to the well-matching numerically computed exact outage probability.

The rest of the paper is organized as follows. The system model is detailed in section \ref{sec:sys_mod}. In section \ref{sec:analysis}, a closed-form expression for the outage probability upper bound is derived.  The presented theoretical results are numerically verified in section \ref{sec:results}, with discussions on the potential gains of IGS over conventional PGS. Finally, conclusions are drawn in section \ref{sec:conc}. Lengthy proofs are deferred to the appendices.
%
\section{System Model}\label{sec:sys_mod}
We consider the communication setting depicted in Fig. \ref{sysmodfig}, where a source (${\rm S}$) intends to communicate with a distant destination (${\rm D}$). The direct ${\rm S -  D}$ link is assumed of a relatively weak gain due to path loss and shadowing effects. Accordingly, a full-duplex relay (${\rm R}$) is utilized to assist the end-to-end communication and extend the coverage.  FDR can offer higher spectral efficiency when compared to its half-duplex counterpart. However, FDR in practice suffers from a RSI level which imposes an additional communication challenge. In addition, the received signal component via the ${\rm S -  D}$ link is assumed to be weak and hence, it is considered as interference at the destination. Thus, the FDR system under consideration suffers from two interference sources; the RSI at the relay, and the direct ${\rm S -  D}$ link signal received at the destination. This model has been widely studied for PGS in the literature, and the end-to-end outage probability is derived in closed-form in \cite{kwon2010optimal}.
\begin{figure}[!t]
\centering
\includegraphics[width=0.9\columnwidth]{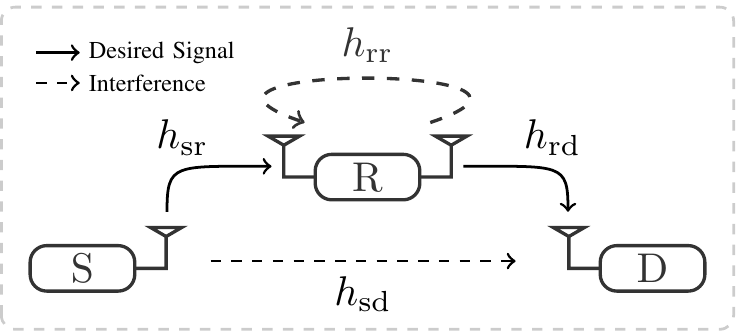}
\caption{A full-duplex cooperative setting in coverage extension scenarios.}
\label{sysmodfig}
\end{figure}
\subsection{Channel Model}\label{subsec:chan_model}
The fading coefficient of the $i-j$ link is denoted by $h_{ij}$, for $i\in\{\rm{s},\rm{r}\}$ and $j\in\{\rm{r},\rm{d}\}$, where $\rm{s}$, $\rm{r}$ and $\rm{d}$ refer to the source, relay, and destination nodes, respectively. Moreover, the $i-j$ link gain is denoted by $g_{ij}=|h_{ij}|^2$, where $\left|.\right|$ is the absolute value. All channels are assumed to follow a block fading model, where $h_{ij}$ remains constant over one block, and varies independently from one block to another following a Rayleigh fading model with average gain $\mathbb{E}\left\{|h_{ij}|^2\right\}=\pi_{ij}$, with $\mathbb{E}\{.\}$ denoting the expectation operator. Accordingly, the channel gain, $g_{ij}$,  is an exponential random variable with mean parameter $\pi_{ij}$. All channel fading gains are assumed to be mutually independent.

The relay operates in a full-duplex mode where simultaneous listening/forwarding is allowed with an introduced level of loopback interference. The link coefficient $h_{\rm{rr}}$ is assumed to represent the RSI after undergoing all possible isolation and cancellation techniques \cite{kwon2010optimal,Riihonen2011Hybrid, khafagy2013outage}. The source and the relay powers are denoted by $P_{\rm{s}}$ and $P_{\rm{r}}$, respectively, where both are restricted to a maximum allowable value of $P_{\mathrm{max}}$. Also, $n_{\rm{r}}$ and $n_{\rm{d}}$ denote the circularly-symmetric complex additive white Gaussian noise components at the relay and the destination, with variance $\sigma_{\rm{r}}^2$ and $\sigma_{\rm{d}}^2$, respectively. Without loss of generality, we assume that $\sigma_{\rm{r}}^2=\sigma_{\rm{d}}^2=1$. 
\subsection{Signal Model}
For ease of exposition, we assume a PGS at the source, $x_{\rm{s}}[t]$, at time $t$. On the other hand, the relay uses a zero-mean IGS, $x_{\mathrm{r}}[t]$, to mitigate the non-negligible RSI at the receiver of the relay. The degree of impropriety of $x_{\mathrm{r}}[t]$ is measured based on the following definitions. 
\begin{definition}\label{def:1}
The variance and pseudo-variance of the relay's transmit signal, $x_{\mathrm{r}}[t]$, are given by $\sigma _x^2={\mathbb{E}} {{{\{\left| x_{\rm{r}} \right|}^2\}}}$ and $\tilde\sigma _x^2=\mathbb{E} {{\{{ x_{\rm{r}}^2 }\}}}$, respectively \rm{\cite{Neeser1993proper}}. 
\end{definition}
\begin{definition}\label{def:2}
A signal is called \emph{proper} if it has a zero $\tilde\sigma _x^2$, while an \emph{improper} signal has a non-zero $\tilde\sigma _x^2$. 
\end{definition}
\begin{definition}\label{def:3}
A circularity coefficient is a measure of the degree of impropriety of the signal $x_{\mathrm{r}}[t]$, which is given as
\begin{equation}\label{C_x}
{{\cal C}_x} = \frac{{\left| {\tilde \sigma _x^2} \right|}}{{\sigma _x^2}}.
\end{equation} 
\end{definition}
\noindent \noindent Following from Definition \ref{def:2}, $\mathcal{C}_x=0$ implies a \emph{proper} signal, while $\mathcal{C}_x=1$ implies a \emph{maximally improper} signal.

 The received signals at the relay and the destination at time $t$ are given, respectively, by
\begin{eqnarray}
y_{\rm{r}}[t]	
&=&
\sqrt{P_{\rm{s}}} h_{\rm{sr}} x_{\rm{s}}[t]+ \sqrt{P_{\rm{r}}} h_{\rm{rr}}x_{\rm{r}}[t] + n_{\rm{r}}[t],\label{relay_rec_signal}\\
y_{\rm{d}}[t]	
&=&
\sqrt{P_{\rm{r}}} h_{\rm{rd}} x_{\rm{r}}[t]+ \sqrt{P_{\rm{s}}} h_{\rm{sd}}x_{\rm{s}}[t] + n_{\rm{d}}[t].\label{dest_rec_signal}
\end{eqnarray}
The relay is assumed to adopt a decode-and-forward relaying strategy, where it does not transmit any message of its own, but forwards the regenerated source message after decoding. Due to the source and relay asynchronous transmissions, the signal transmitted by the relay (source) is considered as an additional noise term at the relay (destination) in the decoding stage as commonly treated in the related literature \cite{kwon2010optimal,Riihonen2011Hybrid}. 
\subsection{Achievable Rates}
From the adopted signal model in \eqref{relay_rec_signal} and \eqref{dest_rec_signal}, each transmit signal, i.e., from the source and the relay transmitter, is considered as a desired signal at one receiver while treated as interference at the other. Hence, the rate expressions for the first and second hops have the same form of those of a two-user interference channel. As a result of using IGS, the achievable rate supported by the ${\rm S -  R}$ link can be expressed ~as \cite{zeng2013transmit}:
\begin{align}\label{R_sr_zeng}
\!\!\!{\mathcal{R}_{\rm{sr}}}\!\left( { {P_{\rm{r}}},{{\cal C}_x}} \right) &={\log _2}\left(\!{1 + \frac{{{P_{\rm{s}}}{g_{\rm{sr}}}}}{{{P_{\rm{r}}}{g_{\rm{rr}}} + 1}}}\!\right)\!+\!\frac{1}{2}{\log _2}\left( {\frac{{1 - {\cal C}_{{{{y}}_{\rm{r}}}}^2}}{{1 - {\cal C}_{{{{I}}_{\rm{r}}}}^2}}} \right),\!\!
\end{align} 
where ${\cal C}_{{{{y}}_{\rm{r}}}}$ and ${\cal C}_{{{{I}}_{\rm{r}}}}$ are the circularity coefficients of the received and
interference-plus-noise signals at the relay, respectively, which are given by
\begin{align}
{{\cal C}_{{{{y}}_{\rm{r}}}}} = \frac{{{P_{\rm{r}}}{g_{\rm{rr}}}{{\cal C}_x}}}{{{P_{\rm{s}}}{g_{\rm{sr}}} + {P_{\rm{r}}}{g_{\rm{rr}}} + 1}},\;\;
{{\cal C}_{{{I}_{\rm{r}}}}} = \frac{{{P_{\rm{r}}}{g_{\rm{r}}}{{\cal C}_x}}}{{{P_{\rm{r}}}{g_{\rm{rr}}} + 1}}.
\end{align}
Hence, \eqref{R_sr_zeng} can be simplified as
\begin{align}\label{R_sr_zeng_simplified}
\!\!\!\!\!\!{\mathcal{R}_{\rm{sr}}}\!\left( { {P_{\rm{r}}},{{\cal C}_x}} \right) \!&=\!\frac{1}{2} {\log _2}\!\left(\!{\frac{{{{\left( {{P_{\rm{s}}}{g_{\rm{sr}}}+ {P_{\rm{r}}}{g_{\rm{\rm{rr}}}} + 1} \right)}^2}\!-\!{{\left( {{P_{\rm{r}}}{g_{\rm{rr}}}{{\cal C}_x}} \right)}^2}}}{{{{\left( {{P_{\rm{r}}}{g_{\rm{rr}}} + 1} \right)}^2}\!-\!{{\left( {{P_{\rm{r}}}{g_{\rm{rr}}}{{\cal C}_x}} \right)}^2}}}}\!\right)\!.\!\!\!
\end{align}
Similarly, the achievable rate supported by the ${\rm R -  D}$ link is given by
\begin{align}\label{R_rd_zeng_simplified}
\!\!\!\!\!\!{\mathcal{R}_{\rm{rd}}}\!\left( { {P_{\rm{r}}},{{\cal C}_x}} \right) \!&=\!\frac{1}{2}{\log _2}\!\left(\!{\frac{{{{\left( {{P_{\rm{r}}}{g_{\rm{rd}}}\!+\!{P_{\rm{s}}}{g_{\rm{sd}}} + 1} \right)}^2}\!-\!{{\left( {{P_{\rm{r}}}{g_{\rm{rd}}}{{\cal C}_x}} \right)}^2}}}{{{{\left( {{P_{\rm{s}}}{g_{\rm{sd}}} + 1} \right)}^2}}}}\!\right)\!.\!\!\!
\end{align}    
One can observe that if $\mathcal{C}_x=0$, we obtain the well known expressions of the achievable rates of PGS.
\section{Outage Performance Analysis}\label{sec:analysis}
In this section, we analyze the outage performance of the canonical cooperative setting depicted in Fig. \ref{sysmodfig} when IGS is allowed at the relay. The end-to-end outage probability is given by
\begin{eqnarray}\label{p_out_overall}
\Pout &=& 1- \Psrbar~\Prdbar,
\end{eqnarray}
where $\Psr$ and $\Prd$ denote the outage probability in the ${\rm S - R}$ and the ${\rm R - D}$ links, respectively, while $\overline{\mathcal{P}}_{ij}=1-\mathcal{P}_{ij}$. In what follows, we derive the outage probability expressions in the individual links, i.e., $\Psr$ and $\Prd$.
\subsection{Outage Probability of ${\rm S - R}$ Link}
 Let $R$ denote the target rate of the ${\rm S - R}$ link, then its outage probability is defined as
 \begin{equation}\label{p_out_sr_prob}
\Psr\left( { {P_{\rm{r}}},{{\cal C}_x}} \right) = \mathbb{P}\left\{ {{\mathcal{R}_{\rm{sr}}}\left( { {P_{\rm{r}}},{{\cal C}_x}} \right) < R} \right\},
 \end{equation}
where $\mathbb{P}\left\{A\right\}$ denotes the probability of occurence of the event $A$. To this end, the {\rm S - R} link outage probability can be obtained from Lemma \ref{lemma1}.
\begin{lemma}\label{lemma1}
In an FDR cooperative system with IGS employed at the relay, the outage probability of the ${\rm S - R}$ link with a target rate $R$ is given by 
\begin{align}
\!\!\!\Psr\left( {{P_{\rm{r}}},{{\cal C}_x}} \right) 
& = 1 - \frac{1}{{{\pi _{\rm{rr}}}}}\int\limits_0^\infty  {{e^{ - \Big( \frac{x}{{{\pi _{\rm{rr}}}}}+\frac{{\left( {{P_{\rm{r}}}x + 1} \right)}}{{{P_{\rm{s}}}{\pi _{\rm{sr}}}}}\Psi \left( {\frac{{{P_{\rm{r}}}x{{\cal C}_x}}}{{{P_{\rm{r}}}x + 1}}} \right)\Big)}}} dx,\label{P_exact_integral_IGS_first_hop}
\end{align}
\rm{where}
\begin{equation}
\Psi \left( x \right) =  {\sqrt {1 + \gamma \left( {1 - {x^2}} \right)}  - 1},
\end{equation}
and $\gamma = {{2^{2R}} - 1}$. 
\end{lemma}
\begin{proof}
By substituting \eqref{R_sr_zeng_simplified} in  \eqref{p_out_sr_prob}, we get
\begin{align}\label{p_out_inequality}
\Psr\left( { {P_{\rm{r}}},{{\cal C}_x}} \right) = &  \mathbb{P}\Big\{ P_{\rm s}^2g_{\rm{sr}}^2 + 2{P_{\rm{s}}}{g_{\rm{sr}}}\left( {{P_{\rm{r}}}{g_{\rm{rr}}} + 1} \right) - \nonumber \\
& \gamma \left( {{{\left( {{P_{\rm{r}}}{g_{\rm{rr}}} + 1} \right)}^2} - {{ {{P_{\rm{r}}^2}{g_{\rm{rr}}^2}{{\cal C}_x^2}} }}} \right) < 0 \Big\},
\end{align}
Hence, ${{\cal P}_{\rm{sr}}}$, conditioned on ${g_{\rm{rr}}}$, can be calculated as
\begin{equation}\label{p_out_sr_conditioned}
\Psr\left( { {P_{\rm{r}}},{{\cal C}_x}\left| {{g_{\rm{rr}}}} \right.} \right) = \frac{1}{{{\pi _{\rm{sr}}}}}\int\limits_0^{g_{\rm{sr}}^\circ } {{{\mathop{e}\nolimits} ^{ - \frac{x}{{{\pi _{\rm{sr}}}}}}}} dx = 1 - {e^{ - \frac{{g_{\rm{sr}}^\circ }}{{{\pi _{\rm{sr}}}}}}},
\end{equation}
where ${g_{\rm{sr}}^\circ }$ is the non-negative zero obtained by solving the inequality in \eqref{p_out_inequality} with respect to ${g_{\rm{sr}}}$ which can be written as
\begin{equation}
g_{\rm{sr}}^\circ  = \frac{\left( {{P_{\rm{r}}}{g_{\rm{rr}}} + 1} \right)}{P_{\rm{s}}}{\Psi}\left( {\frac{{{P_{\rm{r}}}{g_{\rm{rr}}}{{\cal C}_x}}}{{{P_{\rm{r}}}{g_{\rm{rr}}} + 1}}} \right),
\end{equation}
Therefore, by averaging over the statistics of $g_{\rm{rr}}$ in \eqref{p_out_sr_conditioned}, we obtain
\begin{equation}\label{p_out_sr_expectation}
\Psr\left( {{P_{\rm{r}}},{{\cal C}_x}} \right) = 1 - \mathbb{E}{_{{g_{\rm{rr}}}}}\left\{ {{e^{ - \frac{{\left( {{P_{\rm{r}}}{g_{\rm{rr}}} + 1} \right)}}{{{P_{\rm{s}}}{\pi _{\rm{sr}}}}}\Psi \left( {\frac{{{P_{\rm{r}}}{g_{\rm{rr}}}{{\cal C}_x}}}{{{P_{\rm{r}}}{g_{\rm{rr}}} + 1}}} \right)}}} \right\},
\end{equation}
which directly yields \eqref{P_exact_integral_IGS_first_hop}.
\end{proof}
Unfortunately, there is no closed-form expression for the integral in Lemma \ref{lemma1} except at $\mathcal{C}_x=0$, which gives the known PGS outage probability given in \cite{kwon2010optimal} as
\begin{align}\label{p_out_sr_C_0}
{{\cal P}_{{\rm{sr}}}}\left( {{P_{\rm{r}}},0} \right) = 1 - \frac{{{P_{\rm{s}}}{\pi _{{\rm{sr}}}}{e^{ - \frac{\eta }{{{P_{\rm{s}}}{\pi _{{\rm{sr}}}}}}}}}}{{{P_{\rm{s}}}{\pi _{{\rm{sr}}}}+{P_{\rm{r}}}{\pi _{{\rm{rr}}}}\eta  }},
\end{align}
where $\eta  = {2^R} - 1$. Otherwise, we resort to obtain an upper bound on the outage probability of the ${\rm S -  R}$ link as follows.
\begin{proposition}\label{prop_1}
The exponential term inside the expectation operator in \eqref{p_out_sr_expectation} is a convex function in $g_{\rm{rr}}$.
\end{proposition}
\begin{proof}
The proof is given in Appendix \ref{prop_1_proof}.
\end{proof}
Therefore, the upper bound on the ${\rm S -  R}$ link outage probability is given by the following lemma.
\begin{lemma}\label{lemma2} 
When the FDR is allowed to adopt IGS, the ${\rm S -  R}$ link outage probability in terms of the relay's transmit power and circularity coefficient is upper-bounded by
\begin{align}
{\cal P}_{\rm{sr}}^{\rm{UB}} \left( {{P_{\rm{r}}},{{\cal C}_x}} \right) = 1 - {e^{ - \frac{{ {{P_{\rm{r}}}{\pi _{\rm{rr}}} + 1} }}{{{P_{\rm{s}}}{\pi _{\rm{sr}}}}}\Psi \left( {\frac{{{P_r}{\pi _{rr}}{{\cal C}_x}}}{{{P_{\rm{r}}}{\pi _{\rm{rr}}} + 1}}} \right)}}.
\end{align}
\end{lemma}
\begin{proof}
First, we follow Proposition 1, then by applying Jensen's inequality to the expectation in \eqref{p_out_sr_expectation}, we obtain the given outage probability upper bound.
\end{proof}
\subsection{Outage Probability of ${\rm R - D}$ Link}
The outage probability of the ${\rm R - D}$ link at a target rate $R$ b/sec/Hz is defined as
 \begin{equation}\label{p_out_rd_prob}
\Prd \left( {{P_{\rm{r}}},{{\cal C}_x}} \right) = \mathbb{P}\left\{ {{\mathcal{R}_{\rm{rd}}}} \left( 
{{P_{\rm{r}}},{{\cal C}_x}} \right) < R \right\}.
 \end{equation}
Then, the outage probability of the ${\rm R - D}$ link can be obtained from the following result.
\begin{lemma}\label{lemma3}
In an FDR cooperative system with IGS employed at the relay, the outage probability of the ${\rm R - D}$ link with a target rate of $R$ b/s/Hz is expressed as a function of the relay's transmit power and circularity coefficient as
\begin{equation}
\Prd\left( {{P_{\rm{r}}},{{\cal C}_x}} \right) = 1 - \frac{{{e^{ - \frac{{\Psi \left( {{{\cal C}_x}} \right)}}{{{P_{\rm{r}}}{\pi _{\rm{rd}}}\left( {1 - {\cal C}_x^2} \right)}}}}}}{{{P_{\rm{s}}}{\pi _{\rm{sd}}}\frac{{\Psi \left( {{{\cal C}_x}} \right)}}{{{P_r}{\pi _{\rm{rd}}}\left( {1 - {\cal C}_x^2} \right)}} + 1}}. \label{eq:RD_exact}
\end{equation}
\end{lemma} 
\begin{proof}
Similar to the first hop, after substituting \eqref{R_rd_zeng_simplified} in \eqref{p_out_rd_prob}, we obtain the following inequality
\begin{align}
\Prd \left( {{P_{\rm{r}}},{{\cal C}_x}} \right) =& \mathbb{P}\Big\{ P_{\rm{r}}^2g_{\rm{rd}}^2\left( {1 - {\cal C}_x^2} \right) + 2{P_{\rm{r}}}{g_{\rm{rd}}}\left( {{P_{\rm{s}}}{g_{\rm{sd}}} + 1} \right) - \nonumber \\
& \quad \quad \quad \quad  \quad \quad  \gamma {{\left( {{P_{\rm{s}}}{g_{\rm{sd}}} + 1} \right)}^2} < 0 \Big\}.
\end{align} 
Calculating the non-negative zero of the above inequality, we obtain
\begin{equation}
g_{\rm{rd}}^\circ  = \frac{{\left( {{P_{\rm{s}}}{g_{\rm{sd}}} + 1} \right)}}{{{P_{\rm{r}}}}}\frac{{\Psi \left( {{{\cal C}_x}} \right)}}{{\left( {1 - {\cal C}_x^2} \right)}}.
\end{equation}
We can express $\Prd$, conditioned on ${g_{\rm{sd}}}$, as
\begin{equation}
\Prd\left( {{P_{\rm{r}}},{{\cal C}_x}\left| {{g_{\rm{sd}}}} \right.} \right) = 1 - {e^{ - \frac{{g_{\rm{rd}}^\circ }}{{{\pi _{\rm{rd}}}}}}}.
\end{equation}
By averaging over the exponentially distributed $g_{\rm{sd}}$, we directly obtain \eqref{eq:RD_exact}.
\end{proof}
From Lemma \ref{lemma3}, it can be noticed that, for the PGS case, i.e., $\mathcal{C}_x=0$, Eq. \eqref{eq:RD_exact} yields the known expression for PGS in \cite{kwon2010optimal} as 
\begin{equation}\label{p_out_rd_C_0}
{{\cal P}_{{\rm{rd}}}}\left( {{P_{\rm{r}}},0} \right) = 1 - \frac{{{P_{\rm{r}}}{\pi _{{\rm{rd}}}}{e^{ - \frac{\eta }{{{P_{\rm{r}}}{\pi _{{\rm{rd}}}}}}}}}}{{{P_{\rm{r}}}{\pi _{{\rm{rd}}}} + {P_{\rm{s}}}{\pi _{{\rm{sd}}}}\eta }}.
\end{equation}
Also, for the maximally improper case, i.e., $\mathcal{C}_x=1$, it yields
\begin{equation}
\Prd\left( {{P_{\rm{r}}},1} \right) = \mathop {\lim }\limits_{{{\cal C}_x} \to 1} {{\cal P}_{\rm{rd}}}\left( {{P_{\rm{r}}},{{\cal C}_x}} \right) = 1 - \frac{{{e^{ - \frac{\gamma }{{2{P_{\rm{r}}}{\pi _{\rm{rd}}}}}}}}}{{\frac{{\gamma {P_{\rm{s}}}{\pi _{\rm{sd}}}}}{{2{P_{\rm{r}}}{\pi _{\rm{rd}}}}} + 1}}.
\end{equation} 
\subsection{End-to-End Outage Performance}
For the PGS case, from \eqref{p_out_sr_C_0} and \eqref{p_out_rd_C_0}, we have the exact expression for the end-to-end outage probability \cite{kwon2010optimal} as
\begin{equation}\label{e2e_proper}
{{\cal P}_{{\rm{out}}}}\left( {{P_{\rm{r}}},0} \right) = 1 - \frac{{{P_{\rm{s}}}{P_{\rm{r}}}{\pi _{{\rm{sr}}}}{\pi _{{\rm{rd}}}}{e^{ - \eta \left( {\frac{1}{{{P_{\rm{s}}}{\pi _{{\rm{sr}}}}}} + \frac{1}{{{P_{\rm{r}}}{\pi _{{\rm{rd}}}}}}} \right)}}}}{{\left( {{P_{\rm{s}}}{\pi _{{\rm{sr}}}} + {P_{\rm{r}}}{\pi _{{\rm{rr}}}}\eta } \right)\left( {{P_{\rm{r}}}{\pi _{{\rm{rd}}}} + {P_{\rm{s}}}{\pi _{{\rm{sd}}}}\eta } \right)}}.
\end{equation}
On the other hand, when IGS is used by the relay, the end-to-end upper bound of the outage probability can be obtained from Theorem \ref{theorem_e2e_outage}.
\begin{theorem}\label{theorem_e2e_outage}
In an FDR cooperative system to be used for coverage extension with IGS adopted at the relay while considering the direct link as interference at the destination, the end-to-end outage probability as a function of the relay's transmit power and circularity coefficient can be upper bounded by 
\begin{align}\label{p_out_UB}
\!\!\!\!\Pout^{\rm{UB}}  \left( {{P_{\rm{r}}},{{\cal C}_x}} \right) \!=\! 1 - \frac{{{e^{ -\Big(\frac{{\Psi \left( {{{\cal C}_x}} \right)}}{{{P_{\rm{r}}}{\pi _{\rm{rd}}}\left( {1 - {\cal C}_x^2} \right)}}+ \frac{{{{P_{\rm{r}}}{\pi _{\rm{rr}}} + 1} }}{{{P_{\rm{s}}}{\pi _{\rm{sr}}}}}\Psi \left( {\frac{{{P_{\rm{r}}}{\pi _{\rm{rr}}}{{\cal C}_x}}}{{{P_{\rm{r}}}{\pi _{\rm{rr}}} + 1}}} \right)\Big)}}}}{{{P_{\rm{s}}}{\pi _{\rm{sd}}}\frac{{{\Psi}\left( {{{\cal C}_x}} \right)}}{{{P_{\rm{r}}}{\pi _{\rm{rd}}}\left( {1 - {\cal C}_x^2} \right)}} + 1}}.\!\!
\end{align}
\end{theorem}
\begin{proof}
Based on the derived upper bound and exact expressions of the outage probability for ${\rm S - R}$ and ${\rm R - D}$ links from Lemma \ref{lemma2} and Lemma \ref{lemma3}, respectively, and by direct substitution in \eqref{p_out_overall}, we obtain the result.
\end{proof}
\textit{Asymptiotic Analysis:} 
For maximally IGS, we obtain the upper bound of the end-to-end outage probability from the following corollary.
\begin{corollary}\label{corollary1}
When the relay node in an FDR cooperative system uses maximally IGS, the end-to-end outage probability can be upper-bounded by
\begin{equation}
\mathop {\lim }\limits_{{{\cal C}_x} \to 1} \Pout^{\rm{UB}} = 1 - \frac{{2{P_{\rm{r}}}{\pi _{{\rm{rd}}}}{e^{ - \left( {\frac{\gamma }{{2{P_{\rm{r}}}{\pi _{{\rm{rd}}}}}} + \frac{{\left( {{P_{\rm{r}}}{\pi _{{\rm{rr}}}} + 1} \right)}}{{{\pi _{{\rm{sr}}}}}}{\Psi _s}\left( {\frac{{{P_{\rm{r}}}{\pi _{{\rm{rr}}}}}}{{{P_{\rm{r}}}{\pi _{{\rm{rr}}}} + 1}}} \right)} \right)}}}}{{2{P_{\rm{r}}}{\pi _{{\rm{rd}}}} + \gamma {P_{\rm{s}}}{\pi _{{\rm{sd}}}}}}. 
\end{equation}
\end{corollary}  
In order to evaluate the end-to-end outage probability upper bound performance with respect to RSI when using maximally IGS at the relay transmitter, we state the following theorem.
\begin{theorem}
In the limiting case where $\pi_{\rm{rr}} \rightarrow \infty$ with a fixed relay transmit power $P_{\rm{r}}$, the exact end-to-end outage probability for the PGS case ${{\cal P}_{{\rm{out}}}}\left( {{P_{\rm{r}}},0} \right) \rightarrow 1$, while the upper bound for the end-to-end outage probability for the maximally IGS case $\Pout^{\rm{UB}}\left( {{P_{\rm{r}}},1} \right) \rightarrow {K}$, where

\begin{equation}\label{p_out_pirr_inf}
{K} = 1 - \frac{{2{P_{\rm{r}}}{\pi _{{\rm{rd}}}}{e^{ - \left( {\frac{\gamma }{{{P_{\rm{r}}}{\pi _{{\rm{rd}}}}}} + \frac{\gamma }{{{P_{\rm{s}}}{\pi _{{\rm{sr}}}}}}} \right)}}}}{{2{P_{\rm{r}}}{\pi _{{\rm{rd}}}} + \gamma {P_{\rm{s}}}{\pi _{{\rm{sd}}}}}}.
\end{equation} 
\end{theorem}
\begin{proof}
The theorem is obtained from \eqref{e2e_proper} and Corollary \ref{corollary1} by taking the limit at $\pi_{\rm{rr}} \rightarrow \infty$.
\end{proof}

Interestingly, different from the PGS case, the maximally IGS introduces immunity against high RSI and achieves less outage probability with a constant upper bound \eqref{p_out_pirr_inf}, which depends on the quality of both $\rm{S-R}$ and $\rm{R-D}$ links, in addition to the target rate. 
\subsection{Improper Signaling Optimization}
In this part, we optimize the parameters of the IGS transmit signal in order to minimize the end-to-end outage probability given some boundaries for the optimization variables. We consider two scenarios. First, assuming a fixed relay's transmit power, we optimize the circularity coefficient. Second, we optimize the joint power and circularity coefficient.    
\subsubsection{Circularity Coefficient Optimization}
In order to investigate the merits of IGS over conventional PGS in FDR channels, we aim at finding the optimal circularity coefficient value that minimizes the end-to-end outage probability. Specifically, we aim at solving the following optimization problem:
\begin{align}
\mathop {\min }\limits_{{{\cal C}_x}} \quad  & \Pout^{\rm{UB}}\left( {{P_{\rm{r}}},{{\cal C}_x}} \right) \\
\rm{s.t.} \quad 
& 0 \leq {{\cal C}_x} \le 1. \nonumber
\end{align}
In order to solve the optimization problem, we analyze the convexity properties of the objective function $\Pout^{\rm{UB}}\left( {{P_{\rm{r}}},{{\cal C}_x}}\right)$. In general, the function is found to be non-convex due to the indefinite sign of the second derivative. However, other desirable properties that allow us to find the global optimal point are presented in the following theorem.
\begin{theorem}\label{theorem_unimodality}
When IGS is employed at the relay, the upper bound of the end-to-end outage probability is either a monotonic or a unimodal function in ${\cal C}_x$ over the interior of the region of interest, $0\leq{\cal C}_x\leq1$.
\end{theorem}
\begin{proof}
The proof is provided in Appendix \ref{theorem_unimodality_proof}.
\end{proof}
Since monotonicity and unimodality are special cases of quasi-convexity, such a result allows for the use of quasi-convex optimization algorithms. For instance, the optimal ${{\cal C}_x}$ can be numerically obtained using the well-known bisection method operating on its derivative given in appendix B.
\subsubsection{Joint Power and Circularity Coefficient Optimization}
The power optimization problem in PGS is formulated as
\begin{align}
\mathop {\min }\limits_{{P_{\rm{r}}}} \quad  & {\cal P}_{\rm out}\left( {{P_{\rm{r}}},0} \right) 
\\
s.t. \quad \nonumber
& 0 < {P_{\rm{r}}} \le {P_{\max }}.\nonumber
\end{align}
Also, for the IGS case, the joint problem is given as follows:
\begin{align}\label{improper_opt}
\mathop {\min }\limits_{{P_{\rm{r}}},{{\cal C}_x}} \quad  & {\cal P}_{\rm out}^{\rm UB}\left( {{P_{\rm{r}}},{{\cal C}_x}} \right)  
\\ 
s.t. \quad \nonumber
& 0 < {P_{\rm{r}}} \le {P_{\max }},\\
\quad
& 0 \leq {{\cal C}_x} \le 1. \nonumber
\end{align}
 It can be readily verified that the PGS end-to-end outage prabability function is a non-convex function in the relay power. However, it can be shown that the interior of the function is unimodal in the relay power following similar footsteps of the proof in Theorem \ref{theorem_unimodality}\footnote{Proof is omitted due to space limitations.}. Hence, the bisection method can be used to locate the global optimum.
 
The second problem is a minimization of non-convex objective function with simple box constraints. Thus, one may try to solve it numerically by, for example, the gradient projected method or the projected Newton's method without any guarantee to converge to an optimal solution \cite{Bertsekas1999nonlinear}. For performance analysis purposes, we find the optimal solution via a grid search. Moreover, in Theorem \ref{theorem_unimodality}, we proved that the objective function is either a monotonic or unimodal in the relay's circularity coefficient over the interior of the constraint set. Although it could not be analytically shown, the objective function in \eqref{improper_opt} with a fixed circularity coefficient is observed to exhibit similar properties in the relay power which, if true, makes it possible to be solved by the bisection method to obtain optimal relay power for a given circularity coefficient. This observation motivates us to use a coordinate descent method based on a two-dimensional bisection algorithm as in \cite{fadel2012qos}. Fortunately, as it will be noticed in the numerical results section, it always converges numerically to the optimal solution obtained by exhaustive grid search.        
 
\section{Numerical Results}\label{sec:results}
We numerically evaluate the benefits that can be reaped by employing IGS FDR. Throughout the following, we compare the performance of IGS to that of PGS as a benchmark. For PGS, we show the unoptimized performance with maximum power allocation (MPA), alongside that with optimized relay power using the bisection algorithm (BA) in addition to a fine grid search (GS) for verification purposes. On the other hand, the IGS outage performance is shown via two expressions, namely, a) the derived upper bound (UB) in \eqref{p_out_UB} and, b) the exact end-to-end expression involving the numerical computation of the integral in \eqref{p_out_sr_expectation}. The IGS optimization involves two variables; $P_{\rm{r}}$ and $\mathcal{C}_x$. Hence, we consider two cases for IGS in the presented figures, namely, i) one-dimensional (1D) optimization over $\mathcal{C}_x$ while adopting maximum power allocation for $P_{\rm{r}}$, and ii) joint $P_{\rm{r}}$/$\mathcal{C}_x$ two-dimensional (2D) optimization. The optimization is done for the two aforementioned cases using both BA and GS, with the prefixes 1D and 2D to distinguish between them. We use the following parameters unless otherwise stated: $\pi_{\rm{sr}}=\pi_{\rm{rd}}=20\;\rm{dB}$, $\pi_{\rm{rr}}=10\;\rm{dB}$ and $\pi_{\rm{sd}}=3\;\rm{dB}$. The source and relay maximum power budget $P_{\rm{max}}=1\;W$ and the target rate is $R=1\;\rm{b/s/Hz}$. The source is assumed to use its maximum power budget.\\

\vspace{-4mm}\textit{Effect of RSI for different $\rm{S-R}$ link gains:}    
In Fig. \ref{ICC_EX_1}, we plot the end-to-end outage probability versus $\pi_{\rm{rr}}$ at different $\pi_{\rm{sr}}$ values. As shown, one can observe that at lower values of the RSI, the IGS solution reduces to PGS since the RSI is low and the relay can use more power without deteriorating the $\rm{S-R}$ link quality-of-service. As $\pi_{\rm{rr}}$ increases, the outage performance of the PGS is significantly deteriorated. On the other hand, the IGS design saturates at a fixed level as it can be seen from \eqref{p_out_pirr_inf}. However, this constant value of the outage probability depends on the target rate and the $\rm{S-R}$ and $\rm{R-D}$ link conditions which can be clearly noticed from the outage performance at the two values of $\pi_{\rm{sr}}$. Similar outage performance is observed for different $\pi_{\rm{rd}}$.\\     
\begin{figure}[!t]
\centering
\includegraphics[width=0.9\columnwidth]{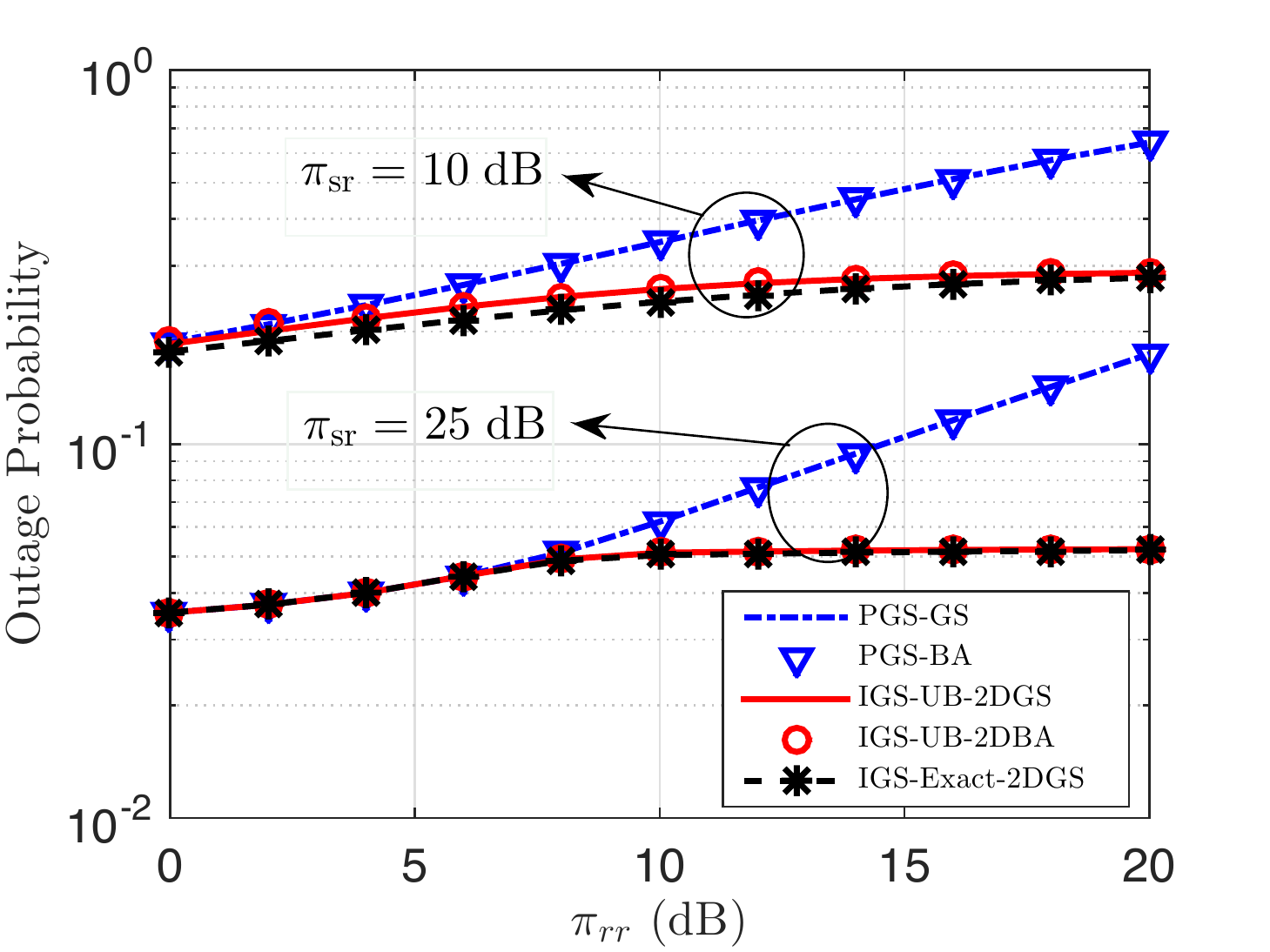}
\vspace{-1mm}
\caption{Outage probability performance vs. $\pi_{\rm{rr}}$ for different $\pi_{\rm{sr}}$ values.}
\label{ICC_EX_1}
\vspace{-3mm}
\end{figure} 

\vspace{-4mm}\textit{Effect of Allowable Relay Power Budget:}    
 In Fig. 4, we study the outage performance versus the available power budget at the relay. For FDR with PGS, and specifically when the relay transmits with its maximum power, the outage performance is known to be enhanced by increasing the allowable power only till a breakeven point as shown. This point is where the increasing adverse effect of RSI on the first hop due to higher relay power starts to exceed any performance returns due to the higher reliability of the second hop. After such a point, any increase in the relay power causes a steady increase in the end-to-end the outage probability. If relay power optimization is allowed in PGS, the performance can at best be kept constant after this breakeven point by clipping the transmit power level, rendering any further increase in the power budget unutilized. On the other hand, the performance trend is different when IGS is adopted at the relay node. Indeed, by optimizing the relay's circularity coefficient, the outage probability continues its decreasing trend. It is also observed that, unlike in PGS, the relay tends to use its maximum power in IGS when joint power/circularity optimization is considered. For high power budgets, the optimal circularity coefficient value approaches unity, denoting a maximally improper signal that tends to allocate most of its power in only one dimension of the complex signal space. This renders the worst case scenario to have the remaining orthogonal signal space dimension as self-interference-free. The decreasing trend of the outage probability in IGS, however, still shows diminishing returns due to the outage performance bottleneck in the first hop, which is primarily influenced by the first hop gain, $\pi_{\rm{sr}}$.
\begin{figure}[!t]
\centering
\includegraphics[width=0.9\columnwidth]{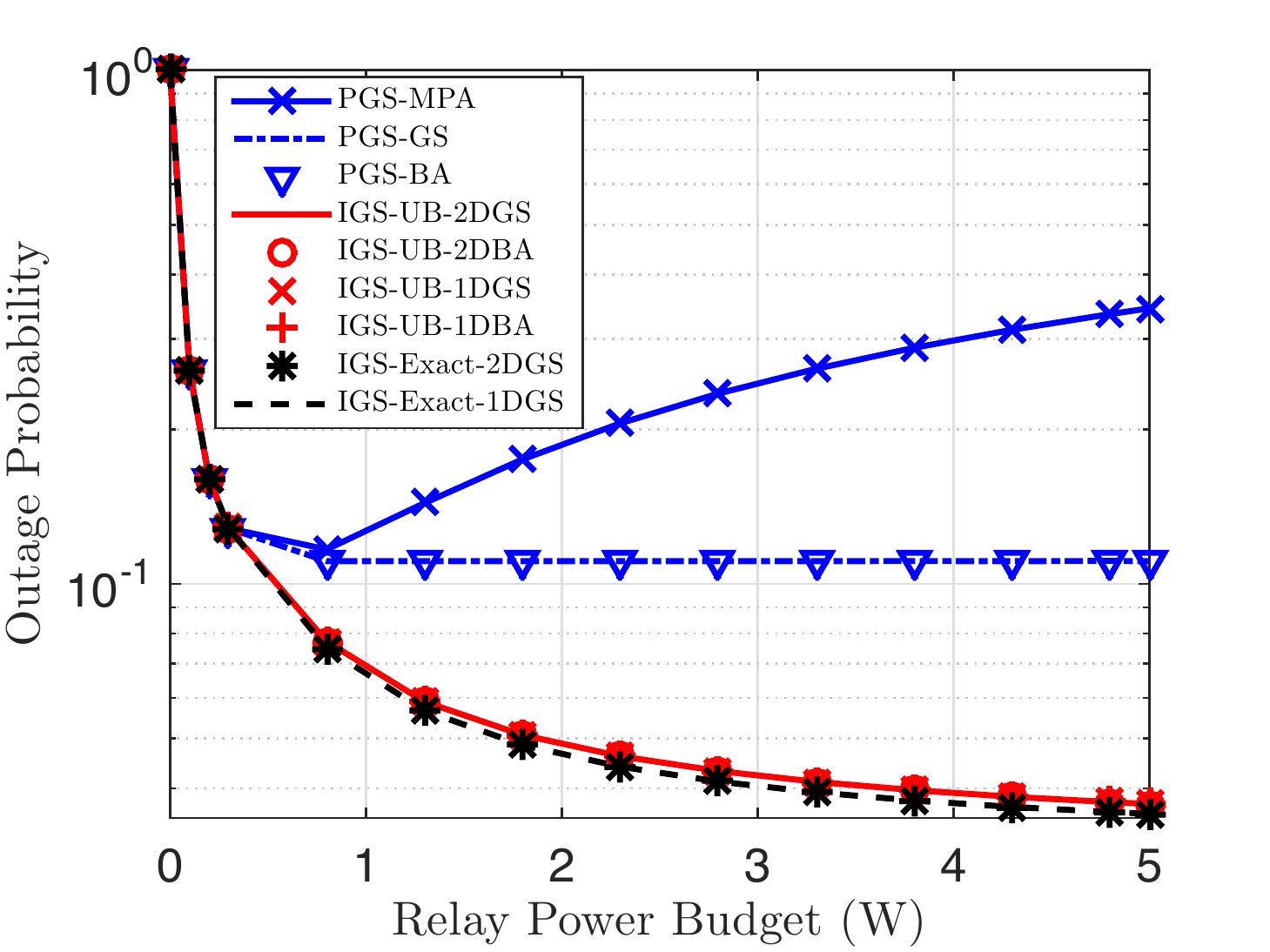}
\caption{Outage probability performance vs. relay power budget.}
\label{ICC_EX_3}
\end{figure} 

\vspace{-1mm}\textit{Effect of average $\rm{S-R}$ link gain:} In Fig. \ref{ICC_EX_2}, we plot the outage probability versus $\pi_{\rm{sr}}$ for different source target rates. First, communication fails at low $\pi_{\rm{sr}}$ values due to the first hop bottleneck, causing the outage probability of both PGS and IGS to start close to unity. As $\pi_{\rm{sr}}$ increases, using IGS enables the relay to utilize more power relative to PGS to boost the performance. At the end, when $\pi_{\rm{sr}}$ reaches a significantly higher value than the RSI, the first hop no longer operates in the interference-limited regime, and hence, the IGS merits become less significant relative to PGS. Finally, as shown, the merits of IGS over PGS are more clear as the target rate decreases. In this case, the rate requirements in the first hop become less stringent, allowing IGS to utilize higher transmit power relative to PGS and yielding a better performance.   
\begin{figure}[!t]
\centering
\includegraphics[width=0.9\columnwidth]{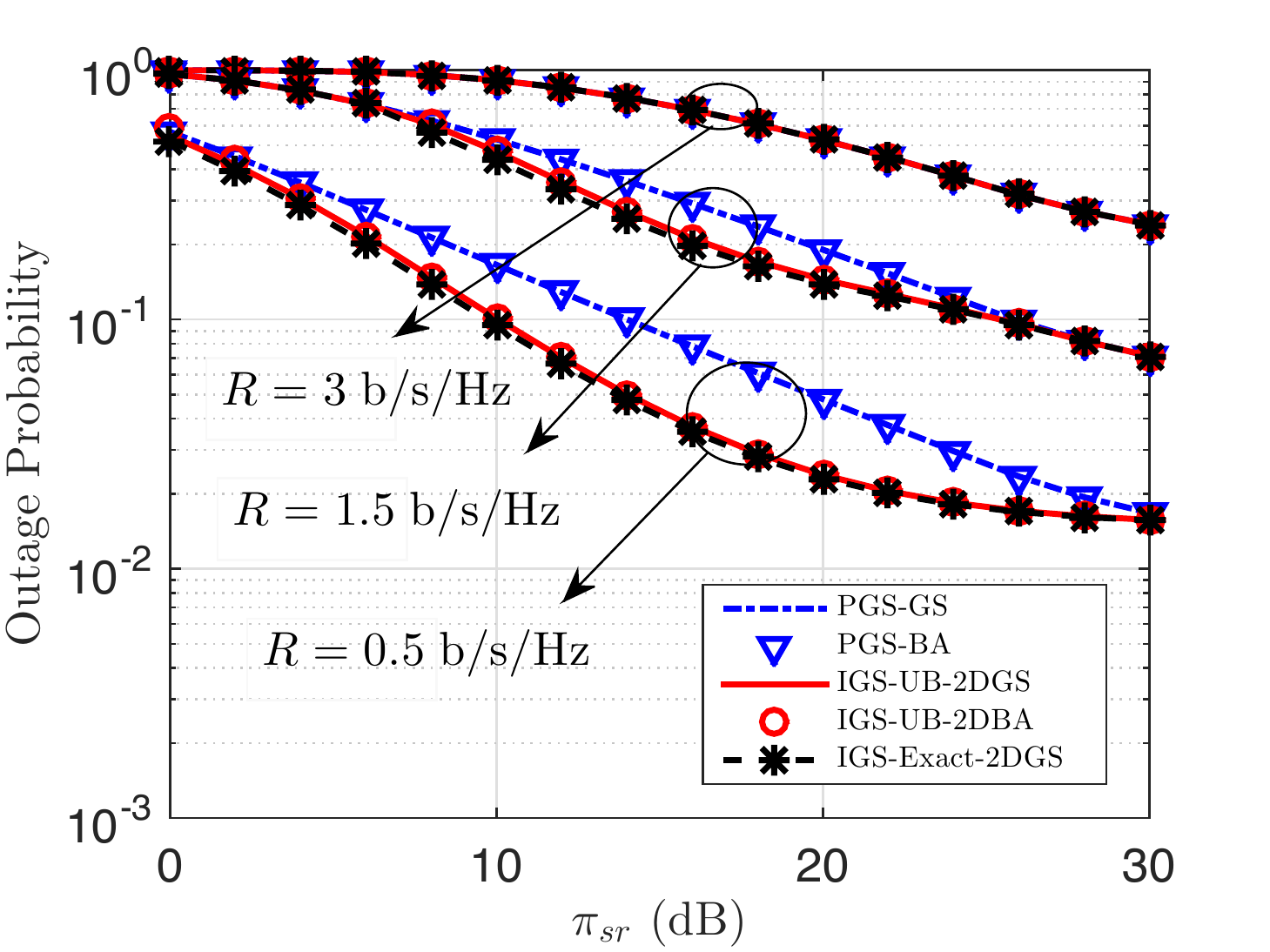}
\caption{Outage probability performance vs. $\pi_{\rm{sr}}$ for different target rates.}
\label{ICC_EX_2}
\vspace{-3mm}
\end{figure} 

\vspace{-1mm}\textit{Effect of Relay Location:}    
We study the relative relay location impact on the end-to-end outage performance for $\pi_{\mathrm{rr}} \in \{0, \; 15\} \; \mathrm{dB}$ and $R=0.5\;\rm{b/s/Hz}$. The relay location in Fig. \ref{ICC_EX_4} is defined as the normalized distance of $\rm{S-R}$ link with respect to the distance of $\rm{S-D}$ link. When the relay location is closer to the source, the $\rm{S-R}$ link gain is very strong relative to the RSI. In such a relatively self-interference-free scenario, the IGS solution reduces as expected to the PGS solution. As the relay moves towards the destination, the relative adverse effect of RSI increases, causing the first hop to operate in the interference-limited regime. In such a regime, the benefits of IGS start to show up in mitigating the adverse effect of the RSI by tuning the signal impropriety. This gives the performance improvement in the second hop, due to the higher $\rm{R-D}$ link gain, a better opportunity to enhance the end-to-end performance. When the relay is too close to the destination, the RSI effect significantly decreases due to the very low relay power required for successful communication, yielding similar IGS/PGS performance. It is clear that the benefits of IGS are noticeable only when the RSI link effect is non-negligible. When $\pi_{\mathrm{rr}}=0$ dB, i.e., at the noise level, IGS yields the PGS solution. 
\begin{figure}[!t]
\centering
\includegraphics[width=0.9\columnwidth]{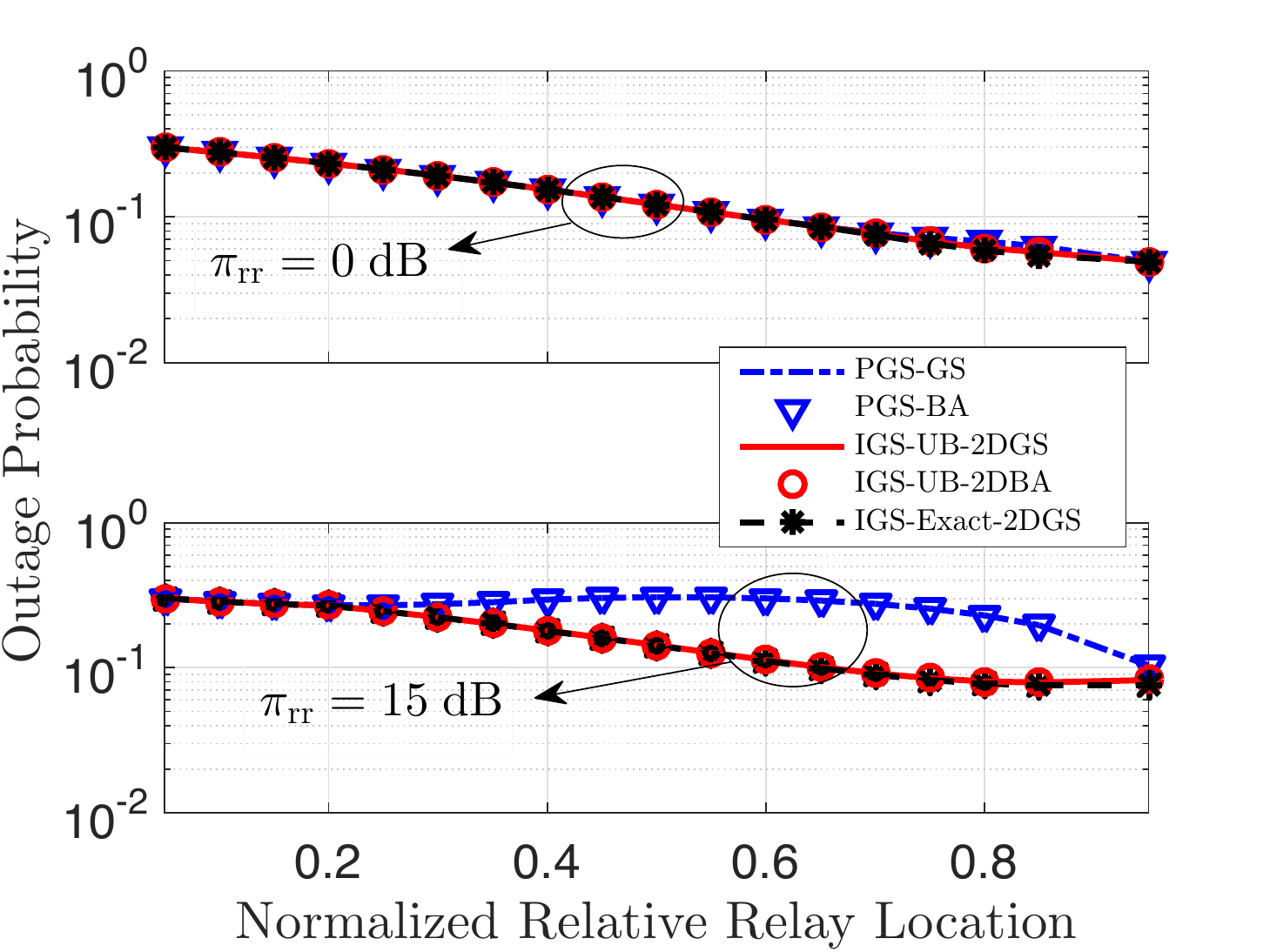}
\caption{Outage probability performance vs. normalized relay location.}
\label{ICC_EX_4}
\end{figure} 
\section{Conclusion}\label{sec:conc}
In this work, we study the potential merits of employing improper Gaussian signaling (IGS) in full-duplex relay channels with non-negligible residual self-interference (RSI). To analyze the benefits of IGS, we derive an upper bound for the end-to-end outage probability. Interestingly, it is shown that IGS offers good immunity against RSI relative to conventional proper Gaussian signaling (PGS). Moreover, we show that, at large RSI values, IGS attains a fixed value that depends on the channel statistics and the target rate. In order to minimize the end-to-end outage probability, we numerically optimize the relay transmit power and circularity coefficient based only on the relay knowledge of the channel statistics. The paper findings show that IGS yields a promising performance over PGS. Specifically, for strong RSI, IGS tends to leverage higher power budgets to enhance the performance, while alleviating the RSI impact by tuning the relay's circularity coefficient.
%
\appendices
\section{Proof of Proposition \ref{prop_1}}\label{prop_1_proof}
We prove the convexity of the exponential term inside the expectation operator in \eqref{p_out_sr_expectation} by expressing it as ${{e^{ - f\left( {{g_{\rm{rr}}}} \right)}}}$. In fact, one can show easily that $f\left( {{g_{\rm{rr}}}} \right)$ can be written as 
\begin{equation}
f\left( {{g_{\rm{rr}}}} \right) = \sqrt {Ag_{\rm{rr}}^2 + B{g_{\rm{rr}}} + C}  - (D{g_{\rm{rr}}} + F).
\end{equation}
where $A=\frac{{P_{\rm{r}}}^2 \left(1+\gamma (1-{{\cal{C}}_x}^2)\right)}{{P_{\rm{s}}}^2 {\pi_{\rm{sr}}}^2}$, $B=\frac{2 (1+\gamma) {P_{\rm{r}}}}{{P_{\rm{s}}}^2 {\pi_{\rm{sr}}}^2}$, $C=\frac{(1+\gamma)}{{P_{\rm{s}}}^2 {\pi_{\rm{sr}}}^2}$, $D=\frac{{P_{\rm{r}}}}{{P_{\rm{s}}} {\pi_{\rm{sr}}}}$, and $F=\frac{1}{{P_{\rm{s}}} {\pi_{\rm{sr}}}}$ are positive. Indeed, the second derivative of $f\left( {{g_{\rm{rr}}}} \right)$ with respect to $g_{\rm{rr}}$ is
\begin{equation}
\frac{{{\partial ^2}f\left( {{g_{\rm{rr}}}} \right)}}{{\partial g_{\rm{rr}}^2}} = \frac{{4AC - {B^2}}}{{4{{\left( {C + {g_{\rm{rr}}}\left( {B + A{g_{\rm{rr}}}} \right)} \right)}^{3/2}}}} \leq 0,
\end{equation}  
since $1+\gamma(1-{{\cal{C}}_x}^2)\leq 1+\gamma$ for ${0\leq {\cal{C}}_x}\leq 1$. Hence, $f\left( {{g_{\rm{rr}}}} \right)$ is concave and ${{e^{ - f\left( {{g_{\rm{rr}}}} \right)}}}$ is convex, which concludes the proof.
%
\section{Proof of Theorem \ref{theorem_unimodality}}\label{theorem_unimodality_proof}
The derived outage probability upper bound as a function of the relay's circularity coefficient is given on the form:
\begin{align}
f(x) = 1 - \frac{{{e^{ - a \frac{{\Psi \left( x \right)}}{{\left( {1 - x^2} \right)}} - b \Psi \left( c x \right)}} }}{{d \frac{{{\Psi}\left( x \right)}}{{\left( {1 - x^2} \right)}} + 1}},
\end{align}
where $0\leq x \leq 1$, $a=\frac{1}{P_{\rm{r}} \pi _{\rm{rd}}}$, $b=\frac{{{{P_{\rm{r}}}{\pi _{\rm{rr}}} + 1} }}{{{P_{\rm{s}}}{\pi _{\rm{sr}}}}}$, $c={\frac{{{P_{\rm{r}}}{\pi _{\rm{rr}}}}}{{{P_{\rm{r}}}{\pi _{\rm{rr}}} + 1}}}$, and $d=\frac{P_{\rm{s}} \pi _{\rm{sd}}}{P_{\rm{r}} \pi _{\rm{rd}}}$. We analyze the stationary points of $\overline{f}(x)=1-f(x)$. Its derivative is given by
\begin{align}
\frac{d \overline{f}(x)}{d x} = x \frac{ e^{-a \frac{\Psi\left(x\right)}{1-x^2}-b \Psi \left( c x \right)}}{\left(d\frac{\Psi\left(x\right)}{1-x^2}+1\right)^2} S(x),
\end{align}
where{
\begin{small}
\begin{eqnarray}
S(x) \!\!\!\!\!\!&=&\!\!\!\!\!\!\left(d\frac{\Psi\left(x\right)}{1-x^2}+1\right)\!\!\left(\!\frac{a \left(2 \Psi\left(x\right)+\gamma 
   \left(x^2-1\right)\right)}{(\Psi\left(x\right)+1) \left(1-x^2\right)^2}+\frac{b \gamma
    c^2}{\Psi \left( c x \right)+1}\!\right)\nonumber\\
&& +\frac{\gamma  d}{(\Psi\left(x\right)+1)
   \left(1-x^2\right)}-\frac{2 d \Psi\left(x\right)}{\left(1-x^2\right)^2}.
   \end{eqnarray}
\end{small}}From the given form, and in addition to the roots of $S(x)$, it is clear that $\frac{d \overline{f}(x)}{d x}$ admits only a zero at $x=0$. Now, we investigate the roots for $S(x)$, and use the change of variables, $z={\Psi}\left( x \right) + 2$. Hence, $1-x^2 = \frac{z (z-2)}{\gamma}$. After substitution and some manipulations, $S(z)$ is hence given for our region of interest, $2 \leq z \leq 1+ \sqrt{1+\gamma}$, by{
\begin{small}\begin{eqnarray}
S(z) =\left(d\frac{\gamma}{z}+1\right) \left(\frac{- a \gamma^2}{z^2 (z-1)}+\frac{b \gamma
    c^2}{\Psi \left( c x \right)+1}\right)-\frac{\gamma^2  d}{z^2 (z-1)}.
   \end{eqnarray}
\end{small}}Since $0<c<1$, we know that $1-c^2 x^2\geq 1-x^2$. Hence, $\Psi \left(c x\right)+1\geq\Psi \left(x\right)+1=z-1$. Let $\Psi \left(c x\right)+1=t_z(z-1)$, where $t_z\geq1$. Therefore,
\begin{eqnarray}
S(z) &=& \frac{(d \gamma+z)(- a \gamma^2 t_z + b \gamma
    c^2 z^2)-\gamma^2  d t_z z}{t_z z^3 (z-1)}.
\end{eqnarray}
The numerator is a cubic polynomial in $z$ which is given by
\begin{eqnarray}
T(z) = b c^2 \gamma z^3     + b c^2 d \gamma^2 z^2 - (a+d) \gamma^2 t_z z - a d \gamma^3 t_z.
\end{eqnarray}
To find the number of positive roots for $T(z)$, we use Descartes rule of signs \cite{prasolov2009polynomials}. Specifically, for the sequence formed by the descending order of the cubic equation coefficients, i.e., the sequence  $\{b c^2 \gamma, b c^2 d \gamma^2, - (a+d) \gamma^2 t_z, - a d \gamma^3 t_z\}$,  the number of sign changes is only one. For our real cubic polynomial, this determines the number of positive roots to be exactly one root. Hence, in the positive region of interest, $2 \leq z \leq 1 + \sqrt{1 + \gamma}$, either one or no \emph{feasible} roots exist for $T(z)$, and hence for $S(z)$. This shows that $\overline{f}(x)$ is either monotonic or unimodal due to the existence of one root at maximum in its interior. If unimodal, the global optimal point can be numerically obtained via the bisection method operating on the derivative function. 

\bibliographystyle{IEEEtran}

\bibliography{IEEEabrv,mgaafar_ref_October_2015}

\end{document}